\documentclass[11pt,letter]{article}
\usepackage{authblk}
\usepackage[margin=1in]{geometry}
\usepackage[utf8x]{inputenc}
\usepackage{amsmath, amsthm}
\usepackage{wrapfig,floatflt,graphicx,amssymb,textcomp,array,amsmath}
\usepackage{enumerate,enumitem}
\usepackage{multirow}
\usepackage{tabularx}
\usepackage{color}
\usepackage{todonotes}
\usepackage[titletoc,title]{appendix}

\newcommand{\etal}{{et~al.}}

\title{Better approximation algorithms for maximum weight internal spanning trees in cubic graphs and claw-free graphs
}

\author{Ahmad Biniaz
}

\affil{School of Computer Science\\University of Windsor\\\texttt{ahmad.biniaz@gmail.com}}

\date{}
\newtheorem{lemma}{Lemma}

\newtheorem{theorem}{Theorem}
\newtheorem{observation}{Observation}
\newtheorem*{problem*}{Problem}
\newtheorem*{claim*}{Claim}
\newtheorem*{invariant*}{Invariant}

\begin{document}
	\maketitle
	\begin{abstract}
		Given a connected vertex-weighted graph $G$, the maximum weight internal spanning tree (MaxwIST) problem asks for a spanning tree of $G$ that maximizes the total weight of internal nodes. This problem is NP-hard and APX-hard, with the currently best known approximation factor $1/2$ (Chen~\etal, Algorithmica 2019). For the case of claw-free graphs, Chen~\etal ~present an involved approximation algorithm with approximation factor $7/12$. They asked whether it is possible to improve these ratios, in particular for claw-free graphs and cubic graphs.
		
		We improve the approximation factors for the MaxwIST problem in cubic graphs and claw-free graphs. For cubic graphs we present an algorithm that computes a spanning tree whose total weight of internal vertices is at least $\frac{3}{4}-\frac{3}{n}$ times the total weight of all vertices, where $n$ is the number of vertices of $G$. 
		This ratio is almost tight for large values of $n$.
		For claw-free graphs of degree at least three, we present an algorithm that computes a spanning tree whose total internal weight is at least $\frac{3}{5}-\frac{1}{n}$ times the total vertex weight. The degree constraint is necessary as this ratio may not be achievable if we allow vertices of degree less than three. 
		
		With the above ratios, we immediately obtain better approximation algorithms with factors $\frac{3}{4}-\epsilon$ and $\frac{3}{5}-\epsilon$ for the MaxwIST problem in cubic graphs and claw-free graphs of degree at least three, for any $\epsilon>0$. 
		In addition to improving the approximation factors, the new algorithms are relatively short compared to that of Chen~\etal. The new algorithms are fairly simple, and employ a variant of the depth-first search algorithm that selects a relatively-large-weight vertex in every branching step. Moreover, the new algorithms take linear time while previous algorithms for similar problem instances are super-linear.
	\end{abstract}
\section{Introduction}

The problems of computing spanning trees with enforced properties have been well studied in the fields of algorithms and graph theory. In the last decades, a number of works have
been devoted to the problem of finding spanning trees with few leaves. Besides its own theoretical importance as a generalization of the Hamiltonian path problem, this problem has applications in the design of cost-efficient communication networks \cite{Salamon2008} and pressure-consistent water supply \cite{Binkele-Raible2013}.

The problem of finding a spanning tree of a given graph having a minimum number of leaves (MinLST) is NP-hard---by a simple reduction from the Hamiltonian path problem---and cannot be approximated within a constant factor unless P = NP \cite{Lu1996}. From an optimization point of view, the MinLST problem is equivalent to the problem of finding a spanning tree
with a maximum number of internal nodes (the MaxIST problem). The MaxIST is NP-hard---again by a reduction form the Hamiltonian path problem---and APX-hard as it does not admit a polynomial-time approximation scheme (PTAS) \cite{Li2014}. Although the MinLST is hard to approximate, there are constant-factor approximations algorithms for the MaxIST problem with successively improved approximation ratios $1/2$ \cite{Prieto2003, Salamon2008}, $4/7$ \cite{Salamon2009} (for graphs of degree at least two), $3/5$ \cite{Knauer2015}, $2/3$ \cite{Li2017a}, $3/4$ \cite{Li2014}, and $13/17$ \cite{Chen2018}. The parameterized version of the MaxIST problem, with the number of internal vertices as the parameter, has also been extensively studied, see e.g. \cite{Li2015, Fomin2012, Fomin2013, Binkele-Raible2013, Prieto2003}.

This paper addresses the weighted version of the MaxIST problem. Let $G$ be a connected vertex-weighted graph where each vertex $v$ of $G$ has a non-negative weight $w(v)$. The {\em maximum weight internal spanning tree} (MaxwIST) problem asks for a spanning tree $T$ of $G$ such that the total weight of internal vertices of $T$ is maximized. 

The MaxwIST has received considerable attention in recent years. Let $n$ and $m$ denote the number of vertices and edges of $G$, respectively, and let $\Delta$ be the maximum vertex-degree in $G$.
Salamon \cite{Salamon2009} designed the first approximation algorithm for this problem. Salamon's algorithm runs in $O(n^4)$ time, has approximation ratio $1/(2\Delta -3)$, and relies on the technique of locally improving an arbitrary spanning tree of $G$. 
By extending the local neighborhood in Salamon’s algorithm, Knauer and Spoerhase \cite{Knauer2015} obtained the first constant-factor approximation algorithm for the problem with ratio $1/(3+\epsilon)$, for any constant $\epsilon > 0$.

Very recently, Chen~\etal~\cite{Chen2019} present an elegant approximation algorithm with ratio $1/2$ for the MaxwIST problem. Their algorithm, which is fairly short and simple, works as follows. Given a vertex-weighted graph $G$, they first assign to each edge $(u,v)$ of $G$ the weight $w(u)+w(v)$. The main ingredient of the algorithm is the observation that the total internal weight of an optimal solution for the MaxwIST problem is at most the weight of a maximum-weight matching in $G$. Based on this, they obtain a maximum-weight matching $M$ and then augment it to a spanning tree $T$ in such a way that the heavier end-vertex of every edge of $M$ is an internal node in $T$. This immediately gives a $1/2$ approximate solution for the MaxwIST problem in $G$. The running time of the algorithm is dominated by the time of computing $M$ which is  $O(n \min\{m \log n, n^2\})$. 

The focus of the present paper is the MaxwIST problem in cubic graphs (3-regular graphs) and claw-free graphs (graphs not containing $K_{1,3}$ as an induced subgraph). We first review some related works for these graph classes, and then we provide a summary of our contributions.
\subsection{Related works on cubic graphs and claw-free graphs}
The famous Hamiltonian path problem is a special case of the MaxIST problem that seeks a spanning tree with $n-2$ internal nodes, where $n$ is the total number of vertices. The Hamiltonian path problem is NP-hard in cubic graphs~\cite{Garey1976} and in line graphs (which are claw-free) \cite{Bertossi1981}. Therefore, the MaxIST problem (and consequently the MaxwIST problem) is NP-hard in both cubic graphs and claw-free graphs.


Salamon and Wiener \cite{Salamon2008} present approximation algorithms with ratios $2/3$ and $5/6$ for the MaxIST problem in claw-free graphs and cubic graphs, respectively. 
Although the first ratio has been improved to $13/17$ (even for arbitrary graphs) \cite{Chen2018}, the ratio $5/6$ is still the best known for cubic graphs.
Binkele-Raible~\etal~\cite{Binkele-Raible2013} studied the parameterized version of the MaxIST problem in cubic graphs. They design an FPT-algorithm that decides in $O^*(2.1364^k)$ time whether a cubic graph has a spanning tree with at least $k$
internal vertices. The Hamiltonian path problem (which is a special case with $k=n-2$) arises in computer graphics
in the context of stripification of triangulated surface models \cite{Arkin1996, Gopi2004}. Eppstein studied the problem of counting Hamiltonian cycles and the traveling salesman problem in cubic graphs \cite{Eppstein2007}. 

When restricted to claw-free graphs, Salamon~\cite{Salamon2009} presented an approximation algorithm with ratio $1/2$ for the MaxwIST problem (this is obtained by adding more local improvement rules to their general $1/(2\Delta -3)$-approximation algorithm). In particular, they show that if a claw-free graph has degree at least two, then one can obtain a spanning tree whose total internal weight is at least $1/2$ times the total vertex weight. Chen~\etal~\cite{Chen2019} improved this approximation ratio to $7/12$; they obtain this ratio by extending their own simple $1/2$-approximation algorithm for general graphs. In contrast to their first algorithm which is simple, this new algorithm (restricted to claw-free graphs) is highly involved and relies on twenty-five pages of detailed case analysis. 

For the MaxwIST problem in cubic graphs, no ratio better than $1/2$ (which holds for general graphs) is known. Chen~\etal~\cite{Chen2019} asked explicitly whether it is possible to obtain better approximation algorithms for the MaxwIST problem in cubic graphs and claw-free graphs.  
\subsection{Our contributions}
\label{contribution-section}
We study the MaxwIST problem in cubic graphs and claw-free graphs. We obtain approximation algorithms with better factors for both graph classes. For cubic graphs we present an algorithm (say A1) that achieves a tree whose total internal weight is at least $\frac{3}{4}-\frac{3}{n}$ times the total vertex weight. 
This ratio (with respect to the total vertex weight) is almost tight if $n$ is sufficiently large.

For claw-free graphs of degree at least three we present an algorithm (say A2) that achieves a tree whose total internal weight is at least $\frac{3}{5}-\frac{1}{n}$ times the total vertex weight. This ratio (with respect to the total vertex weight) may not be achievable if we drop the degree constraint.

With the above ratios, immediately we obtain better approximation algorithms with factors $\frac{3}{4}-\epsilon$ and $\frac{3}{5}-\epsilon$ for the MaxwIST problem in cubic graphs and claw free graphs (of degree at least three) as follows. Consider an instance of the MaxwIST problem in a cubic graph with $n$ vertices, and consider any fixed $\epsilon>0$. If $n\leqslant \frac{3}{\epsilon}$, then we solve this instance optimally, say by checking all possible spanning trees. If $n> \frac{3}{\epsilon}$ then we run our algorithm A1. This establishes the approximation factor $\frac{3}{4}-\epsilon$ for cubic graphs. The approximation factor $\frac{3}{5}-\epsilon$ for claw-free graphs is obtained analogously by running A2 instead. As we will see later, this approximation factor holds  even if we allow the presence of degree-1 vertices in the claw-free graph. 

\vspace{10pt}
\noindent{\bf Simplicity.} In addition to improving the approximation ratios, the new algorithms (A1 and A2) are relatively short compared to that of Chen~\etal~\cite{Chen2019}.
The new algorithms are not complicated either. They involve an appropriate use of the depth-first search (DFS) algorithm that selects a relatively-large-weight vertex in every branching step.
The new algorithms take linear time, while previous algorithms for similar problem instances are super-linear. 

\vspace{10pt}
\noindent{\bf Approach and comparisons.} 
Our $3/4-\epsilon$ approximation algorithm for weighted cubic graphs is based on a greedy DFS, similar to the $5/6$ approximation algorithm of  Salamon and Wiener \cite{Salamon2008} for unweighted cubic graphs. However, there are major differences between the two algorithms:
(i) The DFS algorithm of \cite{Salamon2008} selects a vertex with minimum number of non-visited neighbors in every branching step. This criteria does not guarantee a good approximation ratio for the weighted version. Our algorithm uses a different branching criteria that depends on the number of non-visited neighbors and the weight of a node. (ii) The ratio $5/6$ is obtained by a counting argument that charges every leaf of the DFS tree to five internal nodes. The counting argument does not work for the weighted version. The weight of a leaf could propagate over many internal nodes, and thus bounding the approximation ratio in the weighted version does not seem straightforward. To establish the ratio $3/4-\epsilon$ we use more powerful ingredients and different type of analysis (Lemma~\ref {path-lemma}).

Our $3/5-\epsilon$ approximation algorithm for weighted claw-free graphs starts by computing a greedy DFS tree $T$ using another branching criteria. This tree is a $1/2$ approximate solution for the MaxwIST problem. To establish better ratios we need to employ stronger techniques. We perform appropriate local improvements on $T$ (by deleting and adding edges) to obtain another tree that achieves our desired ratio $3/5-\epsilon$. The local improvements are performed in a clever way that guarantees linear running time. 
We note that the $1/2$ approximation algorithm of Salamon~\cite{Salamon2009} takes $O(n^4)$ time (which has been improved to $O(n^3)$ by \cite{Knauer2015}). 

\subsection{Preliminaries}
Let $G$ be a connected undirected graph. A {\em DFS-tree} in $G$ is the rooted spanning tree 
that is obtained by running the depth-first search (DFS) algorithm on $G$ from an arbitrary vertex called the {\em root}.
It is well-known that the DFS algorithm classifies the edges of an undirected graph into {\em tree edges} and {\em backward edges}. Backward edges are the non-tree edges of $G$. These edges have the following property that we state in an observation.

\begin{observation}
	\label{non-tree-edge-obs}
	The two end-vertices of every non-tree edge of $G$ belong to the same path in the DFS-tree that starts from the root. In other words, one end-vertex is an ancestor of the other. 
\end{observation}

Let $T$ be a DFS-tree in $G$. For every vertex $v$ we denote by $d_T(v)$ the degree of $v$ in $T$. For every two vertices $u$ and $v$ we denote by $\delta_T(u,v)$ the unique path between $u$ and $v$ in $T$. For every edge $e$ in $G$, we refer to the end-vertex of $e$ that is closer to the root of $T$ by the {\em higher end-vertex}, and refer to other end-vertex of $e$ by the {\em lower end-vertex}.
 
If $G$ is a vertex-weighted graph and $S$ is a subset of vertices of $G$, then we denote the total weight of vertices in $S$ by $w(S)$.

\section{The MaxwIST problem in cubic graphs}

Let $G$ be a connected vertex-weighted cubic graph with vertex set $V$ such that each vertex $v\in V$ has a non-negative weight $w(v)$. For each vertex $v$ let $N(v)$ be the set containing $v$ and its three neighbors. Let $r$ be a vertex of $G$ with minimum $w(N(r))$. Observe that $w(N(r))\leqslant 4w(V)/n$, where $n=|V|$. We employ a greedy version of the DFS algorithm that selects---for the next node of the traversal---a node $x$ that maximizes the ratio $\frac{w(x)}{u(x)}$ where $u(x)$ is the number of non-visited neighbors of $x$. If $u(x)=0$ then the ratio is $+\infty$. Let $T$ be the tree obtained by running this greedy DFS algorithm on $G$ starting from $r$. Notice that $T$ is rooted at $r$. Also notice that, during the DFS algorithm, for every vertex $x$ (with $x\neq r$) we have $u(x)\in\{0,1,2\}$.

\begin{figure}[htb]
	\centering
	\includegraphics[width=.55\linewidth]{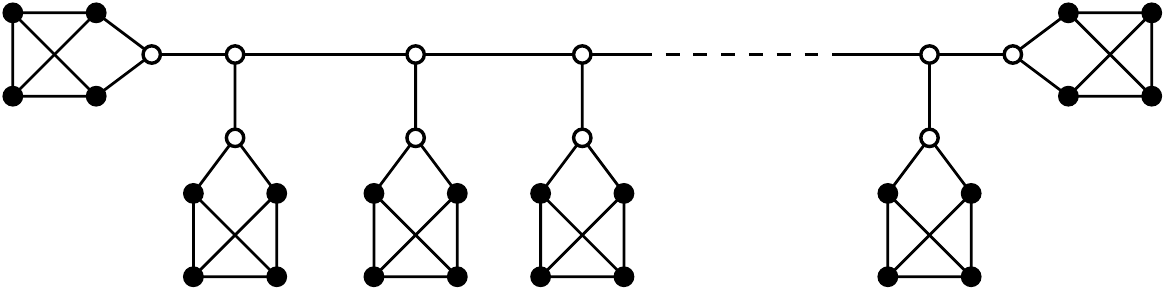}
	\caption{A cubic graph in which black and white vertices have weights 1 and 0, respectively.}
	\label{cubic-fig-0}
\end{figure}

In the rest of this section we show that $T$ is a desired tree, that is, the total weight of internal vertices of $T$ is at least $3/4-3/n$ times the total weight of all vertices. This ratio is almost tight for large $n$. For example consider the cubic graph in Figure~\ref{cubic-fig-0} where every black vertex has weight 1 and every white vertex has weight 0. In any spanning tree of this graph at most three-quarter of black vertices can be internal. 
The following lemma plays an important role in our analysis. In the rest of this section the expression ``while processing $x$'' refers to the moment directly after visiting vertex $x$ and before making decision to follow which of its children in DFS algorithm.

\vspace{10pt}
\noindent\begin{minipage}[t]{\textwidth}
	
	\noindent\begin{minipage}[b]{.82\textwidth}
		\begin{lemma}
			\label{path-lemma}
			Let $x_0x'_1$ be a backward edge with $x_0$ lower than $x'_1$. Let $$\delta(x_0,x'_1)=(x'_1x_1,x_1x'_2,x'_2x_2,\allowbreak\dots,\allowbreak x_{k-1}x'_{k},x'_{k}x_{k})\text{~~~with $k\geqslant 1$ }$$ be a path in $G$ such that each $x'_ix_i$ is a tree-edge where $x_i$ is the child of $x'_i$ on the path $\delta_T(x'_i,x_0)$, each $x_ix'_{i+1}$ is a backward edge with $x_i$ lower than $x'_{i+1}$, and $u(x_k)=2$ while processing $x'_k$. Let $u(x_0)$ be the number of non-visited neighbors of $x_0$ while processing $x'_1$. Then, it holds that $u(x_0)\in\{1,2\}$, and $$w(x_1)+w(x_2)+\dots+w(x_k)\geqslant \frac{{2w(x_0)}}{{u(x_0)}}.$$ Moreover, $d_T(x_i)=d_T(x'_i)=2$ for all $i\in\{1,\dots,k-1\}$ if $k>1$, and either $x'_k$ is the root or it has degree $2$ in $T$.   
		\end{lemma}
	\end{minipage} 
	\hfill
	\begin{minipage}[b]{.17\textwidth}	
		\centering	
		\includegraphics[width=.57in]{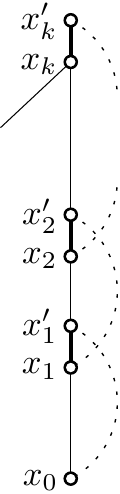}
		\vspace{5pt}
	\end{minipage}
	
\end{minipage}

\begin{proof}
	In the figure to the right, dotted lines represent backward edges, bold-solid lines represent tree-edges, every normal-solid line represents a path in $T$, and the dashed line represents either a tree edge or a backward edge connecting $x_k$ to a lower vertex.
	
	First we verify the potential values of $u(x_0)$. At every step of the DFS algorithm (except for the first step that we choose the root), the number of non-visited neighbors of the next node---to be traversed---is at most $2$ because of the 3-regularity. Since $x_0$ is not the root, $u(x_0)\leqslant 2$. Since $x_0x'_1$ is a backward edge, there is a node, say $x'_0$, on the path $\delta_T(x'_1,x_0)$ such that $x'_0x_0$ is a tree edge (it might be the case that $x'_0=x_1$). Thus at the moment $x'_1$ was processed, $x'_0$ was a non-visited neighbor of $x_0$, and thus $u(x_0)\geqslant 1$. Therefore, $u(x_0)\in\{1,2\}$.
	
	Now we prove the inequality. For every $i\in\{1,\dots, k\}$ it holds that $\frac{w(x_i)}{u(x_i)} \geqslant \frac{w(x_{i-1})}{u(x_{i-1})}$ while processing $x'_i$, because otherwise the greedy DFS would have select $x_{i-1}$ instead of $x_i$. If $k=1$, then by the statement of the lemma we have $u(x_1)=2$ (as $x'_2$ is undefined), and thus $\frac{w(x_1)}{2} \geqslant \frac{w(x_{0})}{u(x_{0})}$ and we are done. Assume that $k\geqslant 2$.
	For every $i\in\{1,\dots, k-1\}$ it holds that $u(x_i)=1$ while processing $x'_i$, because $x_i$ has a visited neighbor $x'_{i+1}$ and a non-visited neighbor which is $x_i$'s child on the path $\delta_T(x_i,x_0)$. For every $i\in\{2,\dots, k\}$ it holds that $u(x_{i-1})=2$ while processing $x'_i$, because $x_{i-1}$ has two non-visited neighbors which are $x'_{i-1}$ and $x_{i-1}$'s child on the path $\delta_T(x_{i-1},x_0)$. By the statement of the lemma $u(x_k)=2$ while processing $x'_k$. Therefore,
	\begin{align}
	\notag w(x_k)&\geqslant w(x_{k-1}),\\
	\notag w(x_1)&\geqslant \frac{{w(x_{0})}}{u(x_0)},~\text{and}\\
	\notag w(x_i)&\geqslant \frac{{w(x_{i-1})}}{2}\text{~for $i\in\{2,\dots,k-1\}$}.
	\end{align}
	The above inequalities imply that
	\begin{align}
	\notag w(x_k)&\geqslant \frac{1}{2^{k-2}}\cdot \frac{{w(x_0)}}{u(x_0)}\text{, and}\\
	\notag w(x_i)&\geqslant \frac{1}{2^{i-1}}\cdot \frac{{w(x_0)}}{u(x_0)}\text{~for $i\in\{1,\dots,k-1\}$}.
	\end{align}
	Therefore,
	\begin{align}
	\notag w(x_1)+w(x_2)+\dots+w(x_k)&\geqslant \left(\frac{1}{2^0}+\frac{1}{2^1}+\dots+\frac{1}{2^{k-3}}+\frac{1}{2^{k-2}}+\frac{1}{2^{k-2}}\right)\cdot \frac{{w(x_0)}}{u(x_0)}\\ \notag &=2\cdot\frac{{w(x_0)}}{u(x_0)}.
	\end{align}  
	
	To verify the degree constraint notice that each vertex $x\in\{x_1,\dots,x_{k-1},x'_1,\dots,x'_{k-1}\}$ has a child and a parent in $T$, and also it is incident to a backward edge. Therefore $d_T(x)=2$. The vertex $x'_k$ has a child in $T$, and also it is incident to a backward edge. If $x'_k$ has a parent in $T$ then $d_T(x'_k)=2$ otherwise it is the root.  
\end{proof}

\begin{wrapfigure}{r}{.99in}
	\vspace{-7pt}
	\centering
	\includegraphics[width=.7in]{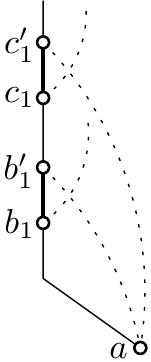}
	\vspace{-3pt}
\end{wrapfigure}
Let $L$ be the set of nodes of $T$ that do not have any children (the leaves); $L$ does not contain the root.
Consider any leaf $a$ in $L$. Let $b'_1$ and $c'_1$ be the higher end-vertices of the two backward edges that are incident to $a$. It is implied from Observation~\ref{non-tree-edge-obs} that both $b'_1$ and $c'_1$ lie on $\delta_T(r,a)$. Thus we can assume, without loss of generality, that $c'_1$ is an ancestor of $b'_1$.

Start from $a$, follow the backward edge $ab'_1$, then follow the tree edge $b'_1b_1$ where $b_1$ is the child of $b'_1$ on $\delta_T(b'_1,a)$, and then follow non-tree and tree edges alternatively and find the path $\delta(a,b'_1)=(b'_1b_1, b_1b'_2,\dots, b'_kb_k)$ with $k\geqslant 1$, that satisfies the conditions of the path $\delta(x_0,x'_1)$ in Lemma~\ref{path-lemma} where $a$ plays the role of $x_0$, $b_i$s play the roles of $x_i$s, and $b'_i$s play the roles of $x'_i$s. (If $u(b_i)<2$ while processing $b'_i$, then $b_i$ must be the lower endpoint of some backward edge $b_i b'_{i+1}$.) Observe that such a path exists and it is uniquely defined by the pair $(a,b'_1)$ because $u(b_k)=2$ while processing $b'_k$ and $d_T(b_i)=d_T(b'_i)=2$ for all $i\in\{1,\dots,k-1\}$ if $k>1$. 
While processing $b'_1$ we have $u(a)=1$. 
Therefore, by Lemma~\ref{path-lemma} we get 
$$w(b_1)+\dots+w(b_k)\geqslant 2w(a).$$

Analogously, find the path $\delta(a,c'_1)=(c'_1c_1,c_1c'_2\dots, c'_lc_l)$ with $l\geqslant 1$, by following the backward edge $ac'_1$. 
Since $u(a)=2$ while processing $c'_1$, Lemma~\ref{path-lemma} implies that  
$$w(c_1)+\dots+w(c_l)\geqslant w(a).$$
Adding these two inequalities, we get 
\begin{equation}
\label{one-leaf-weight}
w(b_1)+\dots+w(b_k)+w(c_1)+\dots+w(c_l)\geqslant 3 w(a).
\end{equation}  
Consider the sets $\{b_1,\dots,b_k\}$ and $\{c_1,\dots,c_l\}$ for all leaves in $L$; notice that there are $2|L|$ sets. All elements of these sets are internal vertices of $T$ as they have a parent and a child. 
If the root is incident to at most one backward edge, then these sets do not share any vertex (because every vertex in these sets has degree 2 in $T$). 
If the root is incident to two backward edges then it has only one child which we denote it by $r_c$. In this case the sets can only share $r_c$. Moreover only two sets can share $r_c$ (because of 3-regularity). 

Let $I$ be a set that contains all internal nodes of $T$ except the root. Then $V=I\cup L\cup \{r\}$. Based on the above discussion and Inequality~\eqref{one-leaf-weight} we have 
$$w(I)\geqslant 3 w(L)-w(r_c)=3\left(w(V)-w(I)-w(r)\right)-w(r_c).$$
By rearranging the terms and using the fact that $w(N(r))\leqslant 4w(V)/n$ we have
$$4w(I)\geqslant 3w(V)-3w(r)-w(r_c)\geqslant 3w(V)-3w(N(r))\geqslant 3w(V)-12w(V)/n$$
Dividing both sides by $4w(V)$ gives the desired ratio 
$$\frac{w(I)}{w(V)}\geqslant \frac{3}{4}-\frac{3}{n}.$$
Therefore, $T$ is a desired tree. As discussed in Section~\ref{contribution-section} we obtain a $\left(\frac{3}{4}-\epsilon\right)$-approximation algorithm for the MaxwIST problem in cubic graphs. Because of the 3-regularity, the number of edges of every $n$-vertex cubic graph is $O(n)$. Therefore, the greedy DFS algorithm takes $O(n)$ time.

\begin{theorem}
There exists a linear-time $\left(\frac{3}{4}-\epsilon\right)$-approximation algorithm for the maximum weight internal spanning tree problem in cubic graphs, for any $\epsilon>0$.
\end{theorem}

\section{The MaxwIST problem in claw-free graphs}

Let $G$ be a connected vertex-weighted claw-free graph with vertex set $V$ such that each vertex $v\in V$ is of degree at least 3 and it has a non-negative weight $w(v)$. 

Our algorithm for claw-free graphs is more involved than the simple greedy DFS algorithm for cubic graphs. For cubic graphs we used the DFS-tree directly because we were able to charge the weight of every internal vertex (except the root) to exactly one leaf as every internal vertex is incident to at most one backward edge. However, this is not the case for claw-free graphs---every internal vertex of a DFS-tree can be incident to many backward edges. To overcome this issue, the idea is to first compute a DFS-tree using a different greedy criteria and then modify the tree.

Here we employ a greedy version of the DFS algorithm that selects a maximum-weight vertex in every branching step. Let $T$ be the tree obtained by running this greedy DFS algorithm on $G$ starting from a minimum-weight vertex $r$. Notice that $T$ is rooted at $r$, and $w(r)\leqslant w(V)/n$, where $n=|V|$.
In the rest of this section we modify $T$ to obtain another spanning tree $T'$ whose total internal weight at least $3/5-1/n$ times its total vertex weight. This ratio may not be achievable if we allow vertices of degree less than 3. For example consider the claw-free graph in Figure~\ref{claw-free-fig-0} where every black vertex has weight 1 and every white vertex has weight 0. In any spanning tree of this graph at most half of black vertices can be internal.

\begin{figure}[htb]
	\centering
	\includegraphics[width=.55\linewidth]{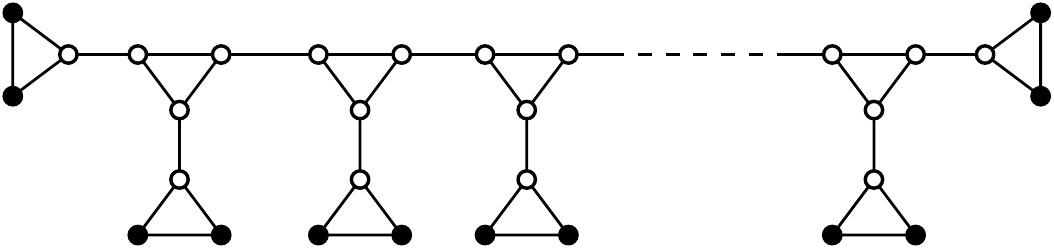}
	\caption{A claw-free graph in which black and white vertices have weights 1 and 0, respectively.}
	\label{claw-free-fig-0}
\end{figure}

\subsection{Preliminaries: some properties of $T$}

The following lemma, though very simple, plays an important role in the design of our algorithm. 

\begin{lemma}
	\label{binary-tree-obs}
	The tree $T$ is a binary tree, i.e., every node of $T$ has at most two children. 
\end{lemma}
\begin{proof}
	If a node $v\in T$ has more than two children, say $v_1,v_2,v_3,\dots$, then by Observation~\ref{non-tree-edge-obs} there are no edges between $v_1$, $v_2$, and $v_3$ in $G$. Therefore, the subgraph of $G$ that is induced by $\{v,v_1,v_2,v_3\}$ is a $K_{1,3}$. This contradicts the fact that $G$ is claw-free. 
\end{proof}

\vspace{2pt}
\begin{wrapfigure}{r}{2in}
		\vspace{-5pt}
		\centering
		\includegraphics[width=1.8in]{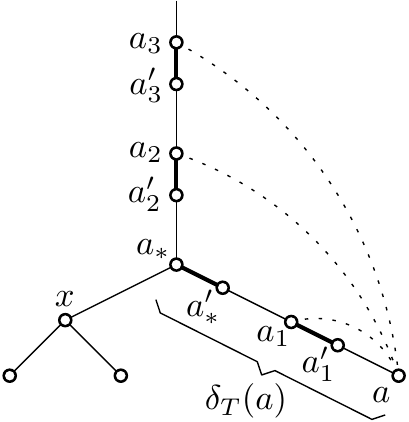}
		\vspace{0pt}
	\end{wrapfigure}
In the following description, ``a leaf of $T$'' refers to a node of $T$ that does not have any children, so the root is not a leaf even if it has degree 1.  
For every leaf $a$ of $T$, we denote by $a_1, a_2,\dots, a_k$ the higher end-vertices of the backward edges incident to $a$ while walking up $T$ from $a$ to the root. Since each vertex of $G$ has degree at least three, every leaf $a$ is incident to at least two backward edges, and thus $k\geqslant 2$. For each $i\in\{1,\dots,k\}$ we denote by $a'_i$ the child of $a_i$ on the path $\delta_T(a_i,a)$; by Observation~\ref{non-tree-edge-obs} such a path exists (it might be the case that $a_i=a'_{i+1}$ for some $i$).
Our greedy choice in the DFS algorithm implies that for each $i$ 
$$w(a'_i) \geqslant w(a).$$
In the figure to the right, dotted lines represent backward edges, bold-solid lines represent tree-edges, and every normal-solid line represents a path in $T$. 
By Lemma~\ref{binary-tree-obs}, $T$ is a binary tree and thus its vertices have degrees $1$, $2$, and $3$. 
For every leaf $a$ of $T$, we denote by $a_*$ the degree-$3$ vertex of $T$ that is closest to $a$, and by $a'_*$ the child of $a_*$ on the path $\delta_T(a_*,a)$; it might be the case that $a'_*=a$. If such a degree-3 vertex does not exist ($T$ is a path), then we set $a_*$ to be the root $r$.
We refer to the path $\delta_T(a_*,a)$ as the {\em leaf-branch} of $a$, and denote it by $\delta_T(a)$ (because this path is uniquely defined by $a$). A leaf-branch is {\em short} if it contains only two vertices $a$ and $a_*$, and it is {\em long} otherwise. A {\em deep branching-vertex} is a degree-3 vertex of $T$ that is incident to two leaf-branches. Thus the subtree of every deep branching-vertex has exactly two leaves. In the figure to the right $x$ is a deep branching-vertex, but $a_*$ is not.
The following lemma comes in handy for the design of our algorithm.

\begin{lemma}
	\label{two-backward-lemma}
	Every internal node $($of $T$$)$ is adjacent to at most two leaves $($of $T$$)$ in $G$.
\end{lemma}
\begin{proof}
	For the sake of contradiction assume that an internal node $v$ is adjacent to more than two leaves, say $l_1, l_2, l_3,\dots$ in $G$. By Observation~\ref{non-tree-edge-obs} there are no edges between $l_1,l_2,l_3$ in $G$. Therefore, the subgraph of $G$ induced by $\{v,l_1,l_2,l_3\}$ is a $K_{1,3}$. This contradicts the fact that $G$ is claw-free. 
\end{proof}

It is implied by Lemma~\ref{two-backward-lemma} that every internal node of $T$ is an end-vertex of at most two backward edges that are incident to leaves of $T$.

\subsection{Obtaining a desired tree from $T$}
We assign to every vertex $v\in V$ a charge equal to the weight of $v$. Thus every vertex $v$ holds the charge $w(v)$. In a general picture, our algorithm distributes the charges of internal nodes between the leaves and at the same time modifies the tree $T$ to obtain another tree $T'$ in which every leaf $a$ (excluding $r$) has at least $2.5w(a)$ charge. The algorithm does not touch the charge of $r$. In the end, $T'$ would be our desired tree. Therefore, if $L'$ is the set of leaves of $T'$ (not including $r$) and $I'$ is the set of internal nodes of $T'$ (again not including $r$) then we have $$w(I')\geqslant 2.5 w(L')-w(L')=1.5w(L')=1.5(w(V)-w(I')-w(r)).$$ 
By rearranging the terms and using the fact that $w(r)\leqslant w(V)/n$ we get 
$$2.5w(I')\geqslant 1.5w(V)-1.5w(V)/n.$$
Dividing both sides by $2.5w(V)$ gives the desired ratio
$$\frac{w(I')}{w(V)}\geqslant \frac{3}{5}-\frac{3}{5n}> \frac{3}{5}-\frac{1}{n}.$$

Now we describe the algorithm in detail. In fact we show how to obtain $T'$ from $T$. Recall $a_1$ and $a_2$ as the first and second end-vertices of backward edges incident to every leaf $a$ of $T$. We start by distributing the charges of internal nodes between the leaves using the following rules. Consider every internal node $q$ of $T$. By Lemma~\ref{two-backward-lemma}, $q$ can be an end-vertex of at most two backward edges that are incident to leaves of $T$.
 
\begin{minipage}{.9\textwidth}
	\vspace{8pt}	
	\centering	
\begin{itemize}[leftmargin=.8in]
	\item [Rule 1.] If there is exactly one leaf $a\in T$ such that $q=a_i$ for some $i\in\{1,2\}$, then transfer the entire charge of $a'_i$ to $a$. See Figure~\ref{rule-fig}.
	\item [Rule 2.] If there are two leaves $a,b\in T$ such that $q=a_i=b_j$ for some $i,j\in\{1,2\}$, then transfer half the charge of $a'_i$ to $a$ and half the charge of $b'_j$ to $b$ (it might be the case that $a'_i=b'_j$, for example if $q$ has only one child). See Figure~\ref{rule-fig}.
\end{itemize}   
	\vspace{2pt}
\end{minipage}
\begin{figure}[htb]
	\centering
	\vspace{-5pt}
	\includegraphics[width=.7\linewidth]{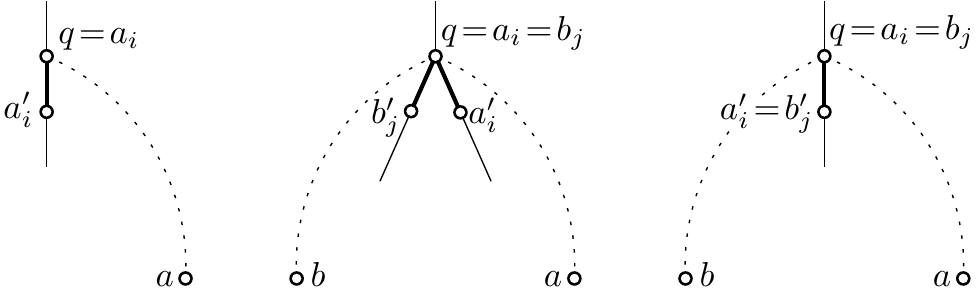}
	\vspace{-3pt}
	\caption{Illustration of Rule 1 and Rule 2.}
		\vspace{-3pt}
	\label{rule-fig}
\end{figure}

If there is no leaf $a\in T$ such that $q=a_1$ or $q=a_2$, then do nothing.
By the above rules, every leaf $a$ receives at least $w(a'_1)/2$ charge from $a'_1$ and at least $w(a'_2)/2$ charge from $a'_2$. Since $w(a'_i)\geqslant w(a)$, the leaf $a$ receives at least $w(a)$ charge. Therefore $a$ holds at least $2w(a)$ charge (including its own charge). 
 To this end, by an analysis similar to the one provided for $T'$ at the beginning of this section, one can conclude that the total internal weight of $T$ is at least $1/2-1/n$ times the total vertex weight. This gives a $(1/2-\epsilon)$-approximation algorithm for the MaxwIST problem by using an argument similar to that of Section~\ref{contribution-section}. 
 
 We say that a leaf $a\in T$ is {\em good} if it holds at least $2.5w(a)$ charge, and {\em bad} otherwise. If all leaves are good then $T$ is a desired tree, and we are done. Assume that $T$ has some bad leaves. Every bad leaf $a$ has the following properties:
\begin{itemize}[leftmargin=.6in]
	\item[(B1)] $a_1$ (and consequently $a_2$) is an ancestor of $a_*$, and
	\item [(B2)] $a'_2$ is equal to $b'_j$ for some leaf $b\neq a$ and some $j\in\{1,2\}$.
\end{itemize}

If property (B1) does not hold then $a'_1$ belongs to the leaf-branch of $a$, and $a'_1\neq a_*$. In this case $a$ receives the entire charge of $a'_1$ by Rule 1. If (B2) does not hold then $a$ receives the entire charge of $a'_2$ again by Rule 1. In either case, $a$ would have at least $2.5w(a)$ charge, and thus it cannot be a bad leaf. Therefore, both (B1) and (B2) hold for every bad leaf $a$. (Property (B2) holds also for $a'_1$ but we do not need this in our analysis.)

Let $T'$ be a copy of $T$.
We go through an iterative process that modifies $T'$ and redistributes charges in such a way that by the end of the process every leaf $a\in T'$ (excluding $r$) holds $2.5w(a)$ charge. This would establish the ratio $3/5-1/n$ as discussed earlier. We emphasis that in the following description $T$ is the original DFS tree and $T'$ is the modified tree during the process.

\vspace{16pt}
\begin{wrapfigure}{r}{2in}
	\vspace{-2pt}
	\centering
	\includegraphics[width=1.6in]{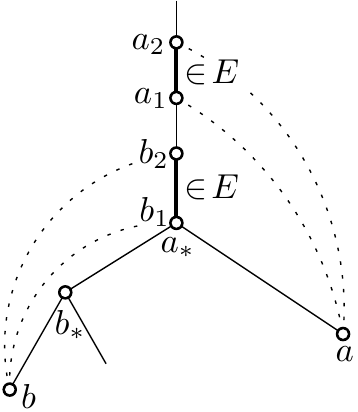}
	\vspace{-2pt}
\end{wrapfigure}
Let $E$ be the empty set. For every leaf $a\in T$ if $a_1$ and $a_2$ are ancestors of $a_*$ and $(a_1,a_2)$ is an edge of $T$ then add $(a_1,a_2)$ to $E$; in this case we say that $a$ {\em introduces} the edge $(a_1,a_2)$ to $E$. If $(a_1,a_2)=(b_1,b_2)$ for some leaf $b\neq a$, then we keep only one instance of this edge in $E$ and choose one of $a$ and $b$ as the introducing vertex; this is consistent with the definition of ``set'' as having no repetitive elements. 

In Section~\ref{process-section} we process the edges of $E$ iteratively. Consider any edge $(a_1,a_2)\in E$ and assume that it is introduced by the leaf $a$. By Rule 1 and Rule 2 the charges of $a'_1$ and $a'_2$ has been transferred to $a$ (and possibly to another leaf). The process will ``release'' at least half the charge of $a'_2$. This released half-charge will be used later in Section~\ref{bad-leaf-section} to take care of bad leaves. During the process, a {\em free charge} refers to the entire or a portion of an internal node's charge that has not been transferred to any leaf, or it was transferred but has been released later. A {\em saturated vertex} is a deep branching-vertex that has lost an incident edge during the process. At the beginning of the process there are no saturated vertices. 

\subsubsection{Processing the edges of $E$}
\label{process-section}
In this section we process the edges of $E$ in such a way that by the end of the process the following properties hold:

\begin{enumerate}[leftmargin=.6in]
	\item[(E1)] For any edge $(a_1,a_2)\in E$, its introducing leaf $a$ is either not a leaf anymore or it has at least $2.5w(a)$ charge.
	\item[(E2)] Any leaf $q$ generated during the process has at least $2.5w(q)$ charge.
	\item[(E3)] All other leaves remain untouched and hold the charges they had before processing $E$. In other word, the process does not convert a good leaf to a bad leaf.
	
	\item[(E4)] For any edge $(a_1,a_2)\in E$ that was processed in case 1 or in case 3 (below), a half-charge of $a'_2$ has been released. We will use these free half-charges to handle bad leaves.
	
	\item[(E5)] Any edge that was removed during the process belongs either to $E$ or to a leaf-branch.
\end{enumerate}
   
Now consider any edge $(a_1,a_2)\in E$ and assume that it is introduced by the leaf $a$. Depending on whether $a_*$ is saturated and/or $a_1$ is the parent of $a_*$ we process $(a_1,a_2)$ in one of the following three cases.

\begin{itemize}[leftmargin=.3in]
	\item[1.] {\em $a_*$ is not saturated.} We consider two sub-cases. 
	\begin{itemize}[leftmargin=.32in]
		\item[1.1.] {\em $\delta_T(a)$ is short.} Remove $(a_1,a_2)$ and $(a,a_*)$ from $T'$ then add $(a,a_1)$ and $(a,a_2)$ as in Figure~\ref{transfer-E}(a). The vertex $a$ is not a leaf anymore. Release any portion (the entire or a half) of charges of $a'_1$ and $a'_2$ that are hold by $a$. If $a_*$ is a deep branching-vertex then it gets saturated (we say that it got saturated by $a$). Notice that if $a$ was holding any portion of $a_*$'s charge (this can happen if $a'_1=a_*$), then this charge is free now.
		\item[1.2.] {\em $\delta_T(a)$ is long.} Let $a'$ be the neighbor of $a$ in $T$; it might be the case that $a'=a'_*$. If $w(a')\geqslant w(a)$ then transfer the charge $w(a')$ from $a'$ to $a$ and release the half-charge of $a'_2$ that is hold by $a$ (although $a$ might hold the entire charge of $a'_2$, the release of a half-charge suffices for the purpose of our algorithm). The vertex $a$ is still a leaf and now it has at least $2.5w(a)$ charge including its own charge; $w(a)+w(a')+0.5w(a'_1)\geqslant 2.5w(a)$. If $w(a')< w(a)$ then remove $(a,a'), (a_1,a_2)$ from $T'$ and add $(a,a_1), (a,a_2)$ as in Figure~\ref{transfer-E}(b). The vertex $a$ is not a leaf anymore, but $a'$ is a new leaf. Transfer the charge $w(a)$ from $a$ to $a'$. Transfer the half-charge of $a'_1$ from $a$ to $a'$, and release the half-charge of $a'_2$ from $a$. The new leaf $a'$ has charge at least $2.5w(a')$ including its own charge; $w(a')+w(a)+0.5w(a'_1)\geqslant 2.5w(a')$. Observe that the charge of $a'$ was not touched by Rule 1 and Rule 2 as $a'$ cannot be an end-vertex of any backward edge incident to a leaf.
	\end{itemize}
	
	\item[2.] {\em $a_*$ is saturated and $a_1$ is the parent of $a_*$.} This case is depicted in Figure~\ref{transfer-E}(c). In this case $a'_1=a_*$. Moreover, either $a$ holds the entire charge of $a_*$, or $a$ holds only a half-charge of $a_*$ and its other half-charge has been released (in case 1.1). In the latter case transfer the free half-charge of $a_*$ to $a$. Now, $a$ holds the entire charge of $a'_1=a_*$. Thus the total charge of $a$ is at least $w(a)+w(a_*)+w(a'_2)\geqslant 2.5w(a)$. This case is ``exceptional'' as the half-charge of $a'_2$ (that is hold by $a$) has not been released.
	  
	\item [3.] {\em $a_*$ is saturated and $a_1$ is not the parent of $a_*$.} In this case, the entire charge of $a_*$ is free (either it was free from the beginning or it was released in case 1.1). We consider two sub-cases.
	\begin{itemize}[leftmargin=.32in]
		\item[3.1.] {\em $\delta_T(a)$ is short.} If $w(a_*)\geqslant w(a)$ then transfer the entire charge of $a_*$ to $a$, and release the half-charge of $a'_2$ from $a$. The vertex $a$ is still a leaf and now it has at least  $w(a)+w(a_*)+0.5w(a'_1)\geqslant 2.5w(a)$ charge. If $w(a_*)< w(a)$ then remove $(a,a_*), (a_1,a_2)$ from $T'$ and add $(a,a_1), (a,a_2)$ as in Figure~\ref{transfer-E}(d). The vertex $a$ is not a leaf anymore, but $a_*$ is a new leaf. Transfer the entire charge of $a$ and the half-charge of $a'_1$ from $a$ to $a_*$, and release the half-charge of $a'_2$ from $a$. The total charge of $a_*$ is at least $w(a)+w(a_*)+0.5w(a'_1)\geqslant 2.5w(a_*)$. 
		\item[3.2.] {\em $\delta_T(a)$ is long.} If $w(a_*)+w(a'_*)\geqslant w(a)$ then transfer the entire charges of $a_*$ and $a'_*$ to $a$ and release the half-charge of $a'_2$ from $a$. The vertex $a$ is still a leaf and now it has at least $w(a)+w(a_*)+w(a'_*)+0.5w(a'_1)\geqslant2.5w(a)$ charge. 
		If $w(a_*)+w(a'_*)< w(a)$ then remove $(a_*,a'_*), (a_1,a_2)$ from $T'$ and add $(a,a_1), (a,a_2)$ as in Figure~\ref{transfer-E}(e). The vertex $a$ is not a leaf anymore, but $a_*$ and $a'_*$ are new leaves. Transfer the entire charge of $a$ and the half-charge of $a'_1$ from $a$ to $a_*$ and $a'_*$ (in such a way that $a_*$ receives at least $1.5w(a_*)$ charge and $a'_*$ receives at least $1.5w(a'_*)$ charge), and release the half-charge of $a'_2$ from $a$. The leaves $a_*$ and $a'_*$ now have at least $2.5w(a_*)$ and $2.5w(a'_*)$ charges respectively (including their own charges). Recall that the entire charge of $a_*$ is free and observe that the charge of $a'_*$ was not touched by Rule 1 and Rule 2. 
	\end{itemize} 	  
\end{itemize}

\begin{figure}[htb]
	\centering
	\setlength{\tabcolsep}{0in}
	$\begin{tabular}{ccccc}
	\multicolumn{1}{m{.2\columnwidth}}{\centering\includegraphics[width=.15\columnwidth]{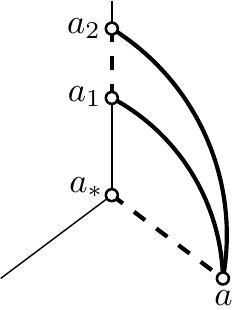}}
	&\multicolumn{1}{m{.2\columnwidth}}{\centering\vspace{0pt}\includegraphics[width=.17\columnwidth]{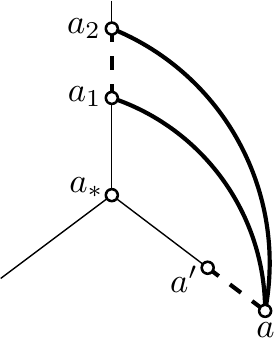}}
	&\multicolumn{1}{m{.2\columnwidth}}{\centering\includegraphics[width=.15\columnwidth]{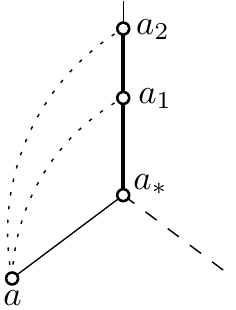}}
	&\multicolumn{1}{m{.2\columnwidth}}{\centering\vspace{0pt}\includegraphics[width=.15\columnwidth]{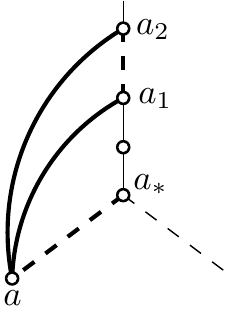}}
	&\multicolumn{1}{m{.2\columnwidth}}{\centering\includegraphics[width=.17\columnwidth]{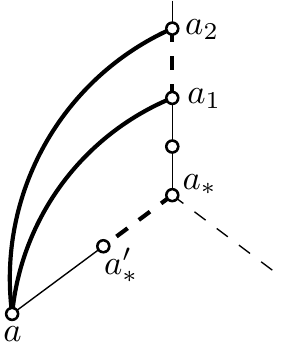}}
	\\
	\vspace{-3pt}
	(a)&(b)&(c)&(d)&(e)
	\vspace{-3pt}
	\end{tabular}$
	\caption{Processing the edge $(a_1,b_1)\in E$ that is introduced by $a$. Bold lines represent edges.}
	\label{transfer-E}
\end{figure}

\subsubsection{Handling bad leaves}
\label{bad-leaf-section}
Consider the tree $T'$ obtained after processing all edges of $E$. By property (E2) any leaf $q\in T'$ that is generated in the above process has at least $2.5w(q)$ charge. By (E1) any bad leaf $a\in T$ that introduced an edge in $E$ is either not a leaf in $T'$ or it has at least $2.5w(a)$ charge. Therefore, these leaves are not bad anymore. In this section we show how to take care of remaining bad leaves, i.e., bad leaves of $T$ that do not introduce any edge in $E$. Consider any such bad leaf $a$ (which has not introduced any edge in $E$). By (E3) the charge of $a$ has not been touched, so $a$ still holds at least $2w(a)$ charge. Recall from (B1) and (B2) that $a_1, a_2$ are ancestors of $a_*$, and $a'_2=b'_j$ for some leaf $b\neq a$ and some $j\in\{1,2\}$. Since $a'_2=b'_j$, they have the same parent in $T$, i.e., $a_2=b_j$. Since $a$, $a'_2$, $b$ are incident to $a_2$ and $G$ is claw-free and $(a,b)\notin G$ (by Observation~\ref{non-tree-edge-obs}), we should have $(a,a'_2)\in G$ or $(b, b'_j)\in G$. We consider two cases.

\begin{itemize}
\item\parbox[t]{\dimexpr\textwidth-\leftmargin}{%
	\vspace{-8pt}
	\begin{wrapfigure}{r}{1.6in}
		\centering
			\vspace{-8pt}
		\includegraphics[width=1.3in]{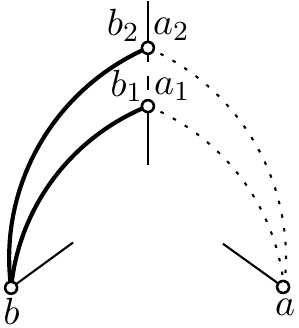}
	\end{wrapfigure}
$(a,a'_2)\in G$. In this case we have $a_1=a'_2$. Thus $(a_1,a_2)\in T$ and consequently $(a_1,a_2)\in E$. Since $a$ has not introduced any edge in $E$, the edge $(a_1,a_2)$ was introduced by a leaf $b\neq a$. Thus $(a_1,a_2)=(b_1,b_2)$ and this edge was processed. We claim that $(b_1,b_2)$ was not processed in the exceptional case 2 (where $b_*$ is a deep branching-vertex that is saturated and $b_1$ is the parent of $b_*$) because otherwise $b_*$ was saturated by its other branch's leaf which is $a$; this contradicts the fact that $a$ has not introduced any edge in $E$. Therefore, $(b_1,b_2)$ was processed either in case 1 or in case 3. By (E4) a half-charge of $b'_2=a'_2$ has been released from $b$. We transfer this free half-charge to $a$, so its total charge is now at least $w(a)+w(a'_2)+0.5w(a'_1)\geqslant 2.5w(a)$. 
}

	\item {\em $(a,a'_2)\notin G$ and $(b, b'_j)\in G$.} In this case $a_1\neq a'_2$. First assume that $j=2$. Then $a_2=b_2$, $b_1=b'_2=a'_2$, $(b_1,b_2)\in E$, and $(b_1,b_2)$ was processed. Also, $b$ is the introducing leaf of $(b_1,b_2)$ because by Observation~\ref{non-tree-edge-obs} the node $b_2$ cannot be incident to any leaf other than $a$ and $b$. Moreover $(b_1,b_2)$ was not processed in case 2 (where $b_*$ is a deep branching-vertex that is saturated) because otherwise $b_*$ was saturated by its other branch's leaf which is $a$; this contradicts the fact that $a$ has not introduced any edge in $E$. Therefore $(b_1,b_2)$ was processed either in case 1 or in case 3, and thus by (E4) a half-charge of $b'_2=a'_2$ is free. We transfer this free half-charge to $a$, so its total charge is now at least $w(a)+w(a'_2)+0.5w(a'_1)\geqslant 2.5w(a)$. 
	
	Now assume that $j=1$. Then $a_2=b_1$, $a'_2=b'_1\neq a_1$, and $(b, b'_1)\in G$. Our choice of $b_1$ and $b_2$ (as the first and second end-vertices of non-tree edges incident to $b$) implies that $(b, b'_1)\in T$. Thus $b'_1=b_*$ and $\delta_T(b)$ is short. Since $a_1\neq b_*$ and $a_1$ is an ancestor of $a_*$ (because $a$ is bad), $a_*$ is in the subtree of $b_*$, and thus $b_*$ is not a deep branching-vertex. We consider two cases.
	
	\begin{itemize}
		\item $(b_1,b_2)\in E$. Then $(b_1,b_2)$ was processed. In this case $b$ is the introducing leaf of $(b_1,b_2)$ because by Observation~\ref{non-tree-edge-obs} the vertex $b_1$ cannot be incident to any leaf other than $a$ and $b$. Since $b_*$ is not a deep branching-vertex (and cannot be saturated) the edge $(b_1,b_2)$ was processed in case 1.1. In this case the charges of  $b'_1=a'_2=b_*$ and $b'_2$ has been released from $b$. We transfer the free half-charge of $b'_1=a'_2$ to $a$, so its total charge is now at least $w(a)+w(a'_2)+0.5w(a'_1)\geqslant 2.5w(a)$.  
		
		\item
		\parbox[t]{\dimexpr\textwidth-2.1\leftmargin}{%
			\vspace{-8pt}
			\begin{wrapfigure}{r}{1.4in}
				\vspace{-8pt}
				\centering
				\includegraphics[width=1.05in]{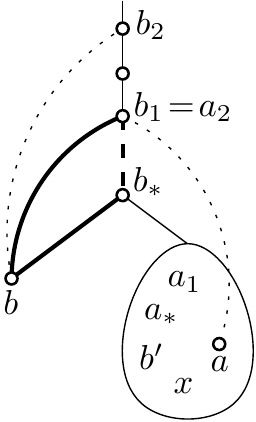}
			\end{wrapfigure} $(b_1,b_2)\notin E$. Then $(b_1,b_2)$ was not processed. The edge $(b_1,b_*)$ is not in $E$ because $b_1$ cannot be incident to any leaf other than $a$ and $b$, and thus by property (E5) this edge is present in $T'$. By the same property, the edge $(b,b_*)$ is also present in $T'$ because it could be removed only if $(b_1,b_2)$ was processed. Let $b'\neq b$ be the other child of $b_*$ in $T$. Either $(b_*,b')$ is present in $T'$, or it was removed and a new edge say $(b_*,x)$ was added for some $x$ in the subtree of $b'$ (this could happen only if $(b_*,b')\in E$). In either case, $b_*$ has degree 3 in $T'$. In this setting we remove $(b_*,b_1)$ from $T'$ and add $(b,b_1)$ as in figure to the right. The vertex $b$ is not a leaf anymore, and no new leaf is generated. Transfer the half-charge of $b'_1=a'_2$ from $b$ to $a$. The leaf $a$ has now at least $w(a)+w(a'_2)+0.5w(a'_1)\geqslant 2.5w(a)$ charge. 
	} 
	\end{itemize}
\end{itemize} 

\subsection{Correctness of the algorithm and final remarks}

Although the correctness of the algorithm should be clear from our construction of $T'$, we briefly describe why $T'$ should be a spanning tree.
This is implied by the fact that we did not ``double count'' any edges: we did not remove any edge twice and did not add any edge twice. By property (E5) any edge that was removed in Section~\ref{process-section} belongs either to $E$ or to a leaf-branch. Every edge in $E$ is introduced by a unique leaf (this is how $E$ was defined), and every edge of a leaf-branch is also introduced by a unique leaf. The edge $(b*,b_1)$ that is removed in the last sub-case of Section~\ref{bad-leaf-section} is uniquely introduced by $b$. Since these introducing leaves have been considered only once (either in Section~\ref{process-section} or in Section~\ref{bad-leaf-section}), the removed edges have not been double counted. Similarly, any edge that was added (in Section~\ref{process-section} or in Section~\ref{bad-leaf-section}) is incident to these introducing leaves. Thus, the added edges have not been double counted either. Therefore, $T'$ is a spanning tree. 
Also, we did not generate any new charges, and only moved the existing charges around. Moreover, by our construction, every leaf $a$ of $T'$ (except $r$) has at least $2.5w(a)$ charge. Therefore, $T'$ is a desired spanning tree, and thus the algorithm is correct.  

\vspace{10pt}
\noindent{\bf Running-time analysis.} To analyze the running time of the algorithm, let $n$ and $m$ be the number of vertices and edges of $G$, respectively. Since $G$ is connected, $n=O(m)$. The greedy DFS algorithm computes the maximum-weight non-visited neighbor of every vertex $v$ at most two times because the DFS-tree is a binary tree (by Lemma~\ref{binary-tree-obs}). Therefore, the time---that DFS algorithm spends at each vertex $v$---is proportional to the number of neighbors of $v$ in $G$. It turns out that the greedy DFS-tree $T$ can be computed in $O(m)$ total time. For every leaf $a\in T$ we can find the vertex $a_*$ by waking up the tree from $a$ until reaching the first degree-3 vertex. Therefore, such degree-3 vertices can be found in $O(n)$ total time for all leaves. To find $a_1$ and $a_2$ we first compute a topological ordering $\mathcal{T}$ of nodes of $T$ in such a way that for any edge of $T$ the child comes before its parent in the ordering. Such an ordering can be computed in $O(n)$ time by walking up the tree from leaves. To determine $a_1$ and $a_2$ we now iterate over all non-tree edges incident to $a$ and find the two whose end-vertices appear before others in $\mathcal{T}$. Therefore, $a_1$ and $a_2$ can be computed in $O(m)$ total time for all leaves $a$. Having $a_1,a_2,a_*$ in hand for every leaf $a$, each edge processing in Section~\ref{process-section} and each bad-leaf handling in Section~\ref{bad-leaf-section} takes $O(1)$ time because operations of this type involve local modification of the tree. 
Therefore the total time of our algorithm is $O(n+m)$. 
  
\vspace{10pt}
\noindent{\bf Inclusion of degree-1 vertices.} As discussed in Section~\ref{contribution-section}, first we obtain a $\left(3/5-\epsilon\right)$ approximation algorithm for the MaxwIST problem in claw-free graphs of degree at least three.
Although the ratio $3/5-1/n$ (with respect to total vertex weight) may not be achievable if we allow vertices of degree less than 3, the approximation factor $3/5-\epsilon$ might be achievable. By a minor modification to our algorithm we can get the same approximation factor even if the graph has degree-1 vertices.
For the objective of the MaxwISP problem we can assume that every degree-1 vertex has weight 0. This assumption is valid as every degree-1 vertex of $G$ will be a leaf in every spanning tree of $G$. During the algorithm we only process leafs that have at least two incident backward edges.
\vspace{8pt}

The following theorem summarizes our result.
\begin{theorem}
	There exists a linear-time $\left(\frac{3}{5}-\epsilon\right)$-approximation algorithm for the maximum weight internal spanning tree problem in claw-free graphs without degree-2 vertices, for any $\epsilon>0$.
\end{theorem}

\section{Conclusions}

Although the ratio $3/4-3/n$ for cubic graphs is almost tight and cannot be improved beyond $3/4$ (with respect to total vertex weight) and the ratio $3/5-1/n$ (with respect to total vertex weight) may not be achievable for claw-free graphs of degree less than 3, approximation factors better than $3/4-\epsilon$ and $3/5-\epsilon$ might be achievable. 
A natural open problem is to improve the approximation factors further. It would be interesting to drop the ``exclusion of degree-2 vertices'' from the $(3/5-\epsilon)$-approximation algorithm for claw-free graphs. Also, it would be interesting to use our greedy DFS technique to obtain better approximation algorithms for the MaxwIST problem in other graph classes.

\bibliographystyle{abbrv}
\bibliography{Max-IST}
\end{document}